\documentclass[a4paper,11pt]{article}
\usepackage{a4wide}

\usepackage{verbatim}
\usepackage{authblk}
\usepackage[utf8]{inputenc}
\usepackage[T1]{fontenc}   
\usepackage[english]{babel} 
\usepackage{amsmath}
\usepackage{amsfonts}
\usepackage{amssymb}
\usepackage{amsthm}
\usepackage{mathrsfs}
\usepackage{mathtools}
\usepackage{array}
\usepackage{url}
\usepackage{enumerate}
\usepackage{graphicx}
\usepackage[all]{xy}
\usepackage{color}
\usepackage{hyperref}
\usepackage{enumerate}
\usepackage{booktabs}
\usepackage{thm-restate}

\newcounter{restate}

\usepackage[linesnumbered, ruled]{algorithm2e}
\newcommand{\inlcomment}[1]{\texttt{\small/* #1 */}}
\newcommand{\eolcomment}[1]{\hfill\texttt{\small// #1}}

\theoremstyle{plain}
\newtheorem{theorem}{Theorem}[section]
\newtheorem{proposition}[theorem]{Proposition}
\newtheorem{lemma}[theorem]{Lemma}

\newtheorem{corollary}[theorem]{Corollary}

{\bfseries}{\itshape}

\theoremstyle{definition}
\newtheorem{definition}{Definition}

\theoremstyle{remark}
\newtheorem{remark}{Remark}

\usepackage{hyperref}
\hypersetup{colorlinks=true, linkcolor=black, citecolor=black, urlcolor=black}

\usepackage{geometry}
\usepackage{fullpage}

\newcommand{\C}{\mathcal{C}}

\newcommand{\OO}{\mathcal{O}}

\newcommand{\PP}{\mathcal{P}}
\newcommand{\F}{\mathbb{F}}

\newcommand{\Pro}{\mathbb{P}}
\newcommand{\Alt}{\mathcal{A}}

\DeclareMathOperator{\GRS}{GRS}
\DeclareMathOperator{\Aut}{Aut}
\DeclareMathOperator{\Perm}{Perm}
\DeclareMathOperator{\PGL}{PGL}
\DeclareMathOperator{\GL}{GL}
\DeclareMathOperator{\Ev}{Ev}
\DeclareMathOperator{\fold}{Fold}
\DeclareMathOperator{\inv}{Inv}

\DeclareMathOperator{\im}{Im}
\DeclareMathOperator{\id}{Id}
\DeclareMathOperator{\Res}{Res}
\DeclareMathOperator{\roots}{roots}

\title{On the security of Some Compact Keys for McEliece Scheme }
\author{\'Elise Barelli\footnote{A short version of this paper was presented at Workshop on Coding and Cryptography (WCC) 2017.}}
\affil{INRIA Saclay and LIX, CNRS UMR 7161 \'Ecole Polytechnique,\\ 91120 Palaiseau Cedex\\email: \href{mailto: elise.barelli@inria.fr}{elise.barelli@inria.fr}}

\begin{document}
\maketitle

\begin{abstract} 
In this paper we study the security of the key of compact McEliece schemes based on 
alternant/Goppa codes with a non-trivial permutation group, in particular quasi-cyclic 
alternant codes. We show that it is possible to reduce the key-recovery problem on the original 
quasi-cyclic code to the same problem on a smaller code derived from the public key. This result is obtained 
thanks to the \textit{invariant} code operation which gives the subcode whose elements are fixed by a
permutation $\sigma \in \text{Perm}(\C)$. The fundamental advantage is that the invariant subcode of an alternant code is an alternant code. This approach improves the technique of Faugère, 
Otmani, Tillich, Perret and Portzamparc which uses \textit{folded} codes of alternant 
codes obtained by using supports globally stable by an 
affine map. We use a simpler 
approach with a unified view on quasi-cyclic alternant codes and we treat the case of automorphisms arising from a non affine homography. In addition, we provide an efficient algorithm to recover the full structure of the alternant code from the structure of the invariant code. 
\end{abstract}

\section{Introduction}

In 1978, McEliece \cite{M78} introduced a public key encryption scheme based on 
linear codes and suggested to use classical Goppa codes which belong to the 
family of alternant codes. This proposition still remains secure but leads to very large 
public keys compared to other public-key cryptosystems. That is why, in despite of its 
fast encryption and decryption, McEliece scheme is limited for practical 
applications. To overcome this limitation, lot of activity devote to decrease the 
key size by choosing codes which admit a very compact public matrix. For instance, 
quasi-cyclic (QC) codes enable to build public key encryption schemes 
with short keys \cite{G05,BCGO09}. These first papers were followed by 
proposals using alternant and Goppa codes with different automorphism groups like 
quasi-dyadic (QD) Goppa codes \cite{BM09}.

The cryptanalisis of code-based schemes can be split in two categories: the message 
recovery attacks and the key-recovery attacks. The first kind of attacks consists in 
generic methods for decoding a random linear code. These methods are known as 
Information Set Decoding (ISD) methods. 
The second kind of attacks consists in recovering the secret elements of the code used 
in the scheme. In this case, the methods are specific to the code family. Recently, 
in the category of key-recovery attacks, new methods appeared, known as \textit{algebraic attacks}. This method consists in recovering the secret element of an alternant code by solving a system of polynomial equations. In \cite{FOPT10}, the authors improved 
this new method to attack QC and QD alternant codes. Such attacks use the specific 
structure of QC/QD codes in order to build an algebraic system with much fewer 
unknowns compared to the generic case. A new approach has been used in \cite
{FOPPT15,FOPPT16} to explain that the reduction of the number of unknowns in the algebraic system comes from a smaller code easily computable from the public generator matrix. This smaller code can
be obtained by summing up the codewords which belong to the same orbit under the 
action of the permutation group and is referred to as the \textit{folded} code.
(We advertise the reader that the folded codes referenced in \cite{G11} are not the
same codes as in this paper.) A relation between the support and multiplier defining the alternant code and those of the folded code exists and is sufficient to find the 
original alternant code. This 
relation comes from the structure of the folded code: \cite{FOPPT16} shows that the 
folding operation preserves the structure of the dual code. That is, the folding of the dual 
of an alternant code (resp. a Goppa code) is the dual of an alternant code (resp. a Goppa code).    

In \cite{FOPPT16}, the authors attack only codes with an automorphism induced by 
an affine transformation acting on the support and the multiplier, we call them 
\textit{affine induced} automorphims. Another kind of quasi-cyclic alternant codes 
can be built from the action of the projective linear group on the support and 
multiplier.
We use in this paper another alternant code built from the public 
generator matrix of a quasi-cyclic alternant code induced by a projective linear transformation, called the \textit{invariant} code and introduced by Loidreau in \cite{L01}.
This invariant code can be built easily from 
the public generator matrix of the alternant code $\C$ since it is the kernel of the 
linear map: $c \in \C \mapsto c - \sigma(c)$, where $\sigma$ is a permutation 
of $\C$.
We remark also that the folded code used by \cite{FOPPT16} is included in the 
invariant code with equality when the characteristic of the field not divide the order of the permutation.

Our main contribution is to consider more general tools coming 
from algebraic geometry and use the invariant code instead of folded code. This 
approach 
has two advantages. First the geometric point of view simplifies the attack by giving a 
unified view of quasi-cyclic alternant codes. It also simplifies some proofs and 
enables to consider alternant codes as algebraic geometric codes on the projective 
line. This method allows us to treat the general
case of projective linear transformations.
The second advantage is that the invariant code operation is applied directly on the 
alternant code and not on the dual code. More precisely, we prove the following results.\\

\textbf{Theorem \ref{thm}.}
\textit{Let $\GRS_k(x,y) \subset \F_{q^m}^n$ be a quasi-cyclic GRS code, and $\sigma \in \mathfrak{S}_n$ of order $\ell$, such that $\ell | n$, the permutation acting on the code $\GRS_k(x,y)$. Then the invariant code $\GRS_k(x,y)^\sigma$ is a GRS code of length $n/\ell$ and dimension $\lfloor k/\ell \rfloor$.}\\

\textbf{Corollary \ref{coro}. }
\textit{Let $\Alt_r(x,y)  \cap \F_{q}^n$ be a quasi-cyclic alternant code, and $\sigma \in \mathfrak{S}_n$ of order $\ell$, such that $\ell |n$, the permutation acting on the code $\Alt_r(x,y)$ . Then the invariant code $\Alt_r(x,y)^\sigma$ is an alternant code of length $n/\ell$ and order $r/\ell$.}\\

In the last section we show that this means that {\bf the key security of compact McEliece scheme 
based on alternant codes with some induced permutation reduces to the key security 
of the invariant code}, which has smaller parameters. We provide an algorithm to recover the secret elements of the alternant code from the knowledge of the its invariant code. This algorithm uses only linear algebra.  

However we can notice that key-recovery is generally more expensive than message recovery. With a
good choice of parameters it is still possible to construct quasi-cyclic codes with high complexity
of key recovery attack on the invariant code.

\section{Quasi-cyclic Alternant Codes}
In this section, we introduce some notation about alternant codes. We denote by $\F_q$ 
the finite field with $q$ elements, where $q$ is a power of a prime $p$. 

Let $x = (x_1,\dots,x_n)$ be an $n$-tuple of distinct elements of $\F_q$, and $y = (y_1,
\dots,y_n)$ be an $n$-tuple of nonzero elements of $\F_q$. The generalized Reed-Solomon (GRS)
code of dimension $k$, denoted GRS$_k(x,y)$, consists of vectors $(y_1f(x_1),\dots,y_nf(x_n))$ where $f$ ranges over all polynomials of degree $ < k$, with 
coefficients in $\F_q$. The vector $x$ is called the \textit{support} and $y$ a \textit{multiplier} of 
the code GRS$_k(x,y)$. In order to define alternant codes, we use the following 
property whose proof can be found in \cite[Chap. 12]{MS86}.
\begin{proposition}
The dual of GRS$_k(x,y)$ is GRS$_{n-k}(x,y^{\perp})$ for some $y^{\perp} \in (\F_q \setminus \{0\} )^n$.
\end{proposition}
\begin{definition}
Let $m$ be a positive integer, $x$ be an $n$-tuple of distinct elements of $\F_{q^m}$ and $y$  be an $n$-tuple of nonzero elements of $\F_{q^m}$. 
The alternant code $\Alt_k(x,y)$ over $\F_q$ is the subfield subcode of GRS$_k(x,y)^{\perp}$, i.e. $\Alt_k(x,y) := \text{GRS}_k(x,y)^{\perp} \cap \F_q^n.$
\end{definition}

\subsection{Representation of $\Alt_k(x,y)$ as a subfield subcode of an AG code}

For the rest of our work, it is convenient to use a projective representation of 
alternant codes. This is possible thanks to algebraic geometric codes introduced by Goppa in 
\cite{G81}. To avoid the confusion with classical Goppa codes which are specific alternant codes, we refer to algebraic geometric codes as AG codes. Any GRS code is an AG code on $\Pro^1$, the projective line over $\F_{q^m}$. Recall some definitions in this case (cf \cite{F69,S93} for further details).

For brevity we denote by $\Pro^1$ the projective line over $\F_{q^m}$. We can consider $\F_{q^m}(\Pro^1)$, the function field over $\F_{q^m}$ associated to the curve $\Pro^1$. We recall that a {\it closed point} of $\Pro^1$ is an orbit of a point, with coordinates in a finite extension of $\F_{q^m}$, under the Frobenius transformation: $(x:y) \mapsto (x^{q^m}:y^{q^m})$. The {\it degree} of a closed point is the cardinality of the orbit. A {\it rational point} is a closed point of degree 1, or equivalently a point whose coordinates are in $\F_{q^m}$. The set of rational points is denoted $\Pro^1(\F_{q^m})$. 
A divisor of $\Pro^1$ is a formal sum, with integers coefficients, of closed points of $\Pro^1$ and for $f \in \F_{q^m}(\Pro^1)$, the principal divisor of $f$, denoted by $(f)$, is defined as the formal sum of zeros and poles of $f$, counted with multiplicity.
For a divisor $G$, we denote by $\deg(G)$ the degree of $G$ and by $L(G) := \{ f \in \F_{q^m}(\Pro^1)~|~(f) \ge -G \}\cup \{0\}$, the Riemann-Roch space associated to $G$.
Let $ \PP = \{ P_1,\dots,P_n \}$ be a set of $n$ distinct points of $\Pro^1$ with coordinates in $\F_{q^m}$ and $G$ be a divisor such that $\deg(G) < n$ and $G$ does not contain any point of $\PP$. We consider the subspace $V \subset \F_{q^m}(\Pro^1)$
of functions without poles in $\PP$ and the following map:
\[ \begin{array}{cccl}
    \Ev_{\PP}: & V & \longrightarrow & \F_{q^m}^n \\
            & f & \longmapsto & (f(P_1),\dots,f(P_n)).
   \end{array} \]

The AG code $C_{L}(\Pro^1,\PP,G)$ is defined by $C_{L}(\Pro^1,\PP,G) := \{ \Ev_{\PP}(f) ~|~f \in L(G) \}.$

Let $x = (x_1,\dots,x_n)$ be an $n$-tuple of distinct elements of $\F_{q^m}$, and $y = (y_1,\dots,y_n)$ be an $n$-tuple of nonzero elements of $\F_{q^m}$. Then GRS$_k(x,y)$ is the 
AG code $C_{L}(\Pro^1,\PP,G)$ where $\PP := \{ (x_i : 1) |~i \in \{1,\dots,n\} \}$ 
and $G := (k-1)P_\infty - (f)$, with $f \in \F_{q^m}(\Pro^1)$ a function with pole order $n-1$ at $P_\infty$, which is the interpolation 
polynomial of degree $n-1$ of $y_1,\dots,y_n$ through the points $x_1,\dots,x_n$. With 
the same notation, we have 
$\Alt_k(x,y) := C_{L}(\Pro^1,\PP,G)^{\perp} \cap \F_q^n.$

\subsection{Induced permutations of Alternant Codes}\label{1}

We explain how we can construct an alternant code invariant under a prescribed 
permutation of the support $\{1,\dots,n\}$, with $n$ the length of the code. In \cite
{D87}, D{\"u}r determines the automorphism group of GRS codes and in \cite{B00,B00a}, 
Berger uses this to construct families of alternant codes invariant under a 
permutation. In particular, Berger deals with some alternant codes invariant under a 
permutation induced by the action of an element of the projective semi-linear group P$\Gamma$L$_2(\F_{q^m})$ on the support and the multiplier. Here we will only be 
interested in projective linear transformations.
First of all, we recall the definition of the projective linear group $\PGL_2(\F_{q^m})$. It is the automorphism group of the projective line $\Pro^1$ defined by:
\[\PGL_2(\F_{q^m}) := \Bigg \{ \begin{array}{ccl}
\Pro^1 & \longrightarrow & \Pro^1\\
(x:y) & \longmapsto & (ax+by:cx+dy)
\end{array}\Big| \begin{cases*}
a,b,c,d \in \F_{q^m},\\
ad-bc \neq 0
\end{cases*}
 \Bigg \}.\]
The elements of $\PGL_2(\F_{q^m})$ have also a matrix representation, i.e.
\begin{equation}\label{sigma}
\forall \sigma \in \PGL_2(\F_{q^m}),\text{ we write }\sigma := \begin{pmatrix}
a&b\\
c&d
\end{pmatrix}, \text{ with } ad-bc \neq 0.
\end{equation}
Where the elements $a,b,c$ and $d$ are defined up to a multiplication by a nonzero scalar. That is to say: 
\[\PGL_2(\F_{q^m}) \simeq \GL_2(\F_{q^m}) / \left\{\begin{pmatrix}
\alpha & 0\\
0 & \alpha
\end{pmatrix}, \alpha \in \F_{q^m}^*
\right\} \] 
Now, we deal with permutations of an alternant code. We recall the following definition.
\begin{definition}
Let $\C$ be a linear code of length $n$ over $\F_{q^m}$. Let $\sigma \in \mathfrak{S}_n$ be a permutation, acting on $\C$ via
$\sigma(c_1,\dots,c_n) = (c_{\sigma(1)},\dots,c_{\sigma(n)}).$
Then the permutation group of a code $\C \subset  \F_{q^m}^n$, is $Perm(\C) := \{\sigma \in \mathfrak{S}_n ~|~ \sigma(\C) = \C \}$.
\end{definition}
In the case of GRS codes, for appropriate dimension, D{\"u}r \cite{D87} shows that the whole permutation group is induced by the action of the projective linear group on the support of the code. The same property has been shown by Stichtenoth \cite{S90}, with the representation of GRS codes as AG rational  codes. More precisely, for appropriate parameters, every permutation of $\C_{L}(\Pro^1, \PP, G)$ is induced by a projective linear transformation. We give the main definitions and theorems of \cite{S90}.

We keep the notation of the previous section. Let $G$ and $G'$ be divisors of $\Pro^1$, we note $G \approx_{\PP} G'$ if there exists $f \in \F_{q^m}(\Pro^1) $, $f \neq 0$, such that $G - G' = (f)$ and $f(P) = 1$, for all $P \in \PP$.
With this definition we have the following lemma:
\begin{lemma}[{\cite{S90}}]
If $G \approx_\PP G'$ then $\C_L(\Pro^1,\PP,G ) = \C_L(\Pro^1,\PP,G')$.
\end{lemma}
Before giving the theorem which allows us to construct any GRS code invariant under a permutation, we define:
\begin{definition}
$\Aut_{\PP,G}(\Pro^1):=\{\sigma \in \Aut(\Pro^1) ~|~ \sigma(\PP) = \PP \text{ and } \sigma(G) \approx_\PP G \}.$
\end{definition} 
\begin{theorem}[{\cite{S90}}]
Let $\C = \C_L(\Pro^1,\PP,G )$ be an AG code with $1 \le deg(G) \le n - 3$. Then $\Perm(\C) = \Aut_{\PP,G}(\Pro^1)$.
\end{theorem}
Now we have all the properties required to construct some alternant codes invariant under a permutation. We consider $\sigma \in \PGL_2(\F_{q^m})$ and $\ell = ord(\sigma)$. We define the support:
\begin{equation}\label{support}
\PP := \coprod\limits_{i=1}^{n/\ell}{Orb_{\sigma}(P_i)},
\end{equation}
where the points $P_i \in \Pro^1(\F_{q^m})$ are pairwise distinct with trivial stabilizer subgroup and $Orb_{\sigma}(P_i) := \{ \sigma^j(P_i) ~|~ j \in \{1..\ell\} \}.$
We define the divisor:
\begin{align}\label{div}
G := \sum\limits_{i=1}^{s}{t_i\sum\limits_{Q \in Orb_\sigma(Q_i)}{Q}},
\end{align}
with $Q_i$ closed points of $ \Pro^1$, $s \in \mathbb{N}$, $t_i \in \mathbb{Z}$ for $i \in \{1,\dots,s\}$ and $\deg(G) = \sum\limits_{i=1}^{s}{t_i}\ell$.

The automorphism $\sigma$ of $\Pro^1$ induces a permutation $\tilde{\sigma}$ of $\C = \C_L(\Pro^1,\PP,G )$ defined by:
\[ \begin{array}{rclc}
\tilde{\sigma} \colon&\C & \longrightarrow & \C\\ 
&(f(P_1),\dots,f(P_n)) &\longmapsto &(f(\sigma(P_1)),\dots,f(\sigma(P_n)))\cdot
\end{array}\]
Then $\tilde{\sigma}$ is also a permutation of $\Alt := \C^{\perp} \cap \F_q^n$. 
For short, we denote by $\sigma \in \PGL_2(\F_{q^m})$ both the homography and the induced permutation on the code $\C$.

\section{Subcodes of Alternant Codes}

We can construct subcodes of $\Alt_r(x,y)$ with smaller parameters, by simple 
operations, which can be used to recover the alternant code $\Alt_r(x,y)$. We 
describe in the next section two subcodes: the \textit{folded} code and the \textit{invariant} code.
Their interactions are also discussed. In the papers \cite{FOPPT15,FOPPT16}, the folding operation 
was used to recover dual of the considered alternant code. Here we do not need to consider the dual code. 
More precisely, we show that for alternant codes with a non trivial permutation group, the invariant code is an alternant code. 

\subsection{Invariant and Folded Codes}

This section deals with subcodes called the \textit{invariant} code and the \textit{folded} code whose definitions are the following.
\begin{definition}\label{def1}
Let $\C$ be a linear code and $\sigma \in Perm(\C)$ of order $\ell$, we consider the following map:
\begin{align*}
\psi \colon & \C \to \C\\ 
&c \mapsto \sum_{i=0}^{\ell-1}{\sigma^i(c)}.
\end{align*}
The folded code of $\C$ is defined by $\fold_\sigma(\C) := \im(\psi) = \im(Id + \sigma + \dots + \sigma^{\ell - 1})$. The invariant code of $\C$ is defined by $\C^{\sigma} := \ker(\sigma - \id)$.
\end{definition}
The folded code was used in \cite{FOPPT15,FOPPT16}, in order to construct a structured  
subcode invariant by a given permutation $\sigma$. Indeed, by the previous definition, 
we remark that $\fold_\sigma(\C)$ is $\sigma$-invariant.
\begin{proposition}
The codes $\fold_\sigma(\C)$ and $\C^{\sigma}$ are subcodes of $\C$ and we have: 
\[ \fold_\sigma(\C) \subseteq \C^\sigma . \] 
\end{proposition}
These two codes are not equal in the general case but we have the following lemma.
We recall that $p$ is the characteristic of $\F_q$ and $\ell = ord(\sigma)$.
\begin{lemma}\label{lem1}
If $p \nmid \ell$ then $\fold_\sigma(\C) = \C^{\sigma}$.
\end{lemma}
\begin{proof}
Let $\psi$ be the map of Definition \ref{def1}, by the previous proposition we know that: 
\[\im(\psi) \subseteq \ker(\sigma - \id).\]
Now, we will show that $\dim(\im(\psi)) = \dim(\ker(\sigma - \id))$.\\
By the rank–nullity theorem we know that: \[\dim(\im(\psi)) =\dim(\C) - \dim(\ker(\psi)).\]
Moreover, $\sigma^{\ell}-\id = (\id+\sigma+\dots+\sigma^{\ell-1})(\sigma-\id)$. Since $p \nmid \ell$, we have: 
\[\gcd(\sum\limits_{i = 0}^{\ell-1}{X^i},X-1) = 1 .\]
Hence, $\C = \ker(\psi) \oplus \ker(\sigma - \id)$
and \[\dim(\ker(\sigma - \id)) = \dim(\C
) - \dim(\ker(\psi)).\]
Therefore $\ker(\sigma - \id) = \im(\psi)$ and $Fold(\C) = \C^{\sigma}$.
\end{proof}

\begin{remark}
In \cite[Example 1]{FOPPT16} an example of the folded and the invariant codes of a 
$\sigma$-invariant alternant code $\Alt$ is given. In this example, the authors 
wrote that $\fold_\sigma(\Alt) \subsetneq \Alt^\sigma$ but in this case these two 
codes must be equal. Indeed, for this example, $p = 3 \nmid 2 = ord(\sigma)$ and
by the previous lemma we have $\fold_\sigma(\Alt) = \Alt^\sigma$. Actually, the computation of the subfield subcode of the folded code in  \cite[Example 1]{FOPPT16} contains a mistake. 
\end{remark}

\begin{remark}
If $c \in \fold_\sigma(\C)$ or $c \in \C^\sigma$, then $c$ takes constant 
value on the orbits under the action of $\sigma$: $\{i,\sigma(i),\dots,\sigma^{\ell-1}(i)\}$. In order to work with codes without repeated coordinates, 
we choose $I \subset \{1,\dots,n\}$ a set of representatives of orbits $\{\sigma^j
(i) | j \in \{0,\dots,\ell-1\} \}$ and we consider the codes punctured on this set:
$\fold_\sigma(\C)|_I$ and $\C^\sigma|_I$. For short, we keep the notations $\fold_\sigma(\C)$ and $\C^\sigma$, for the restricted codes. 
\end{remark}

\begin{remark}
The folding operation is $\F_q$-linear so to apply this operation on a linear code $\C$ it suffices to apply folding operation on a basis of $\C$. This property will be useful in \S \ref{section3.2.3}.
\end{remark}

\subsection{The Invariant Code of $\Alt_r(x,y)$}\label{sec_inv}

In order to study the invariant code of $\Alt_r(x,y)$, we first notice that the invariant operation commutes with the subfield subcode operation.
Indeed, if $\C$ is a linear code over $\F_{q^m}$, $\sigma$-invariant then:
\[
(\C \cap \F_q^n)^\sigma = \{ c \in \C ~|~ c \in \F_q^n \text{ and } \sigma(c) = c\} = \C^\sigma \cap \F_q^n.
\]
In order to prove that the invariant code of $\Alt_r(x,y)$ is also an alternant code we have to prove that the invariant code of a GRS code is a GRS code. Later on, the GRS codes will be described by $\C_{L}(\Pro^1, \PP, G)$, as in Section \ref{1}.
The two lemmata to follow describe the action of an element $\sigma \in \PGL_2$ on the codes $\C_{L}(\Pro^1, \PP, G)$ and provide a description of $\C_{L}(\Pro^1, \PP, G)^\sigma$. 

\begin{lemma}\label{lem4}
Let $c = Ev_\PP(f) \in \C_{L}(\Pro^1, \PP, G)^\sigma$ such that $\sigma(c) = c$, then $f$ is $\sigma$-invariant, i.e. $f\circ \sigma = f$.
\end{lemma}

\begin{proof}
Let $c = (f(P_1),\dots,f(P_n)) \in \C$ such that $\sigma(c) = c$, then:
\begin{align*}
\forall i \in \{1,\dots,n\},~ f(P_{\sigma(i)}) = f(P_i) & \Leftrightarrow \forall i \in \{1,\dots,n\},~ f \circ \sigma(P_i) = f(P_i)\\
& \Leftrightarrow \forall i \in \{1,\dots,n\},~ (f \circ \sigma - f)(P_i) = 0.
\end{align*}
Since $\sigma(G) = G$, $f \circ \sigma \in L(G)$, and then $(f \circ \sigma -f) \in L(G)$. Hence if $(f \circ \sigma -f)$ was nonzero, it should have at most $\deg(G) < n$ zeros on $\Pro^1$, which is a contradiction. Therefore $(f \circ \sigma -f) \equiv 0$ and $f$ is $\sigma$-invariant.
\end{proof}

\begin{lemma}\label{lem5}
Let $\C := \C_L(\Pro^1,\PP,G)$ be an AG code such that $\sigma(\C) = \C$ and $\rho \in 
\PGL_2(\F_{q^m})$. Then $\sigma' := \rho \circ \sigma \circ \rho^{-1}$ induces the same permutation 
on $\C$ as $\sigma$.
\end{lemma}

\begin{proof}
We first prove that: 
\[\C_L(\Pro^1,\rho^{-1}(\PP),\rho^{-1}(G)) = \C_L(\Pro^1,\PP,G).\]
Let $c = (f(P_1),\dots,f(P_n))$ be a codeword of $\C_L(\Pro^1,\PP,G)$. Then, we have $ c = (f \circ \rho \circ \rho^{-1}(P_1),\dots,f \circ \rho \circ \rho^{-1}(P_n)).$
As $f \in L(G)$, the function $h = f \circ \rho \in L(\rho^{-1}(G))$.
Hence, $c \in \{\Ev_{\rho^{-1}(\PP)}(h) ~|~ h \in L(\rho^{-1}(G))\} = \C_L(\Pro^1,\rho^{-1}(\PP),\rho^{-1}(G))$.\\
Now, for all $c = (f(P_1),\dots,f(P_n)) \in \C$, we have:
\begin{align*}
\sigma'(c) & = (f \circ \rho \circ \sigma \circ \rho^{-1}(P_1),\dots,f\circ \rho \circ \sigma \circ \rho^{-1}(P_n))\\
& = (h \circ \sigma (\rho^{-1}(P_1)),\dots, h \circ \sigma (\rho^{-1}(P_n)))
\end{align*}
with $h = f \circ \rho \in L(\rho^{-1}(G))$.
Since $\C_L(\Pro^1,\rho^{-1}(\PP),\rho^{-1}(G)) = \C_L(\Pro^1,\PP,G)$, $\sigma'$ induces the same permutation of the code $\C$ as $\sigma$.
\end{proof}

\begin{theorem}\label{thm}
Let $\C_L(\Pro^1,\PP,G) \subseteq \F_{q^m}^n$ be an AG code of length $n$ and dimension $k$, and $\sigma \in \PGL_2(\F_{q^m})$ of order $\ell$ acting on it, with $\ell | n$. Let $\PP$ and $G$ as in (\ref{support}) and (\ref{div}). Then the invariant code $\C_L(\Pro^1,\PP,G)^\sigma$ is an AG code of length $n/\ell$ and dimension $\lfloor k/\ell \rfloor$.
\end{theorem}

\begin{corollary}\label{coro}
Let $\Alt(\Pro^1,\PP,G) := \C_L(\Pro^1,\PP,G) \cap \F_{q}^n$ be an alternant AG code of length $n$ and order $r$, and $\sigma \in \PGL_2(\F_{q^m})$ of order $\ell$ acting on it, with $\ell |n$. Let $\PP$ and $G$ as in (\ref{support}) and (\ref{div}). Then the invariant code $\Alt(\Pro^1,\PP,G)^\sigma$ is an alternant AG code of length $n/\ell$ and order $\lfloor r/\ell \rfloor$.
\end{corollary}

In order to prove Theorem \ref{thm}, we consider $\sigma \in \PGL_2(\F_{q^m})$ with $\ell = ord(\sigma)$ and we define the support $\PP$ and the divisor $G$ as in (\ref{support}) and (\ref{div}).
Later on, to simplify the proofs we assume that $G$ is constructed from single rational point $Q$, but the result remains true in the general case.

We denote:

\begin{equation}
\begin{array}{cl}
\sigma^j(P_i) & := (\alpha_{i\ell + j}:\beta_{i\ell + j}), \textrm{ for } i \in \{0,\dots,\frac{n}{\ell}-1\}, j \in \{0,\dots,\ell-1\} \\
\sigma^j(Q) & := (\gamma_{j} : \delta_{j}), \textrm{ for } j \in \{0,\dots,\ell-1\}.

\end{array}
\end{equation}

\begin{lemma}\label{lem3}
With the previous notation, we have:
\[L(G) = \Bigg\{ \frac{F(X,Y)}{\prod\limits_{j=0}^{\ell-1}{(\delta_jX-\gamma_jY)^t}}~\big|~ F \in \F_{q^m}[X,Y] \text{ homogeneous polynomial of degree $t\ell$}. \Bigg\}\]

\end{lemma}

To the automorphism $\sigma \in \PGL_2(\F_{q^m})$, we associate a matrix $M := \begin{pmatrix}
a&b \\
c&d
\end{pmatrix}$ as in (\ref{sigma}). Three cases are possibles, depending on the eigenvalues of the matrix $M$: 
\begin{enumerate}
\item $M \sim \begin{pmatrix}
a & 0 \\
0& 1
\end{pmatrix}$, with $a \in \F_{q^m}$ (case diagonalizable in $\F_{q^m}$), 
\item $M \sim \begin{pmatrix}
1& b\\
0&1
\end{pmatrix}$,  with $b \in \F_{q^m}$ (case trigonalizable in $\F_{q^m}$), 
\item $M \sim \begin{pmatrix}
\alpha&0 \\
0&\alpha^q
\end{pmatrix}$, with $\alpha \in \F_{q^{2m}}$ (case diagonalizable in $\F_{q^{2m}}$),
\end{enumerate}
where $M \sim N$, with $M, N \in \PGL_2(\F_{q^m})$, means there exist $P \in \PGL_2(\F_{q^m})$ such that $M = PNP^{-1}$.
In the following, we study these three cases.

\subsubsection{Case $\sigma$ diagonalizable over $\F_{q^m}$}\label{diag}
We suppose $\sigma = \rho \circ \sigma_d \circ \rho^{-1}$ with $\sigma_d$ diagonal 
and $\rho \in \PGL_2(\F_{q^m})$ an automorphism of $\Pro_{\F_{q^m}}^1$. W.l.o.g and 
by Lemma \ref{lem5}, one can assume that:
\begin{equation}
\begin{array}{rclc}\label{sigma1}
\sigma \colon & \Pro^1 &\to &\Pro^1 \\ 
&(x:y)& \mapsto &(ax:y),
\end{array}
\end{equation}
with $a \in \F_{q^m}^*$.

\begin{proposition}\label{prop_inter1}
Let $F \in \F_q[X,Y]$ be a homogeneous polynomial of degree $t\ell$, and $a \in \F_{q^m}$ of order $\ell$.
If $F(aX,Y) = F(X,Y)$, then $F(X,Y) = R(X^\ell,Y^\ell)$, with $R \in \F_{q^m}[X,Y]$ a homogeneous polynomial of degree $t$.
\end{proposition}
A proof is given in \cite[Prop 4]{FOPPT16}. Here, we present a simpler proof of Proposition \ref{prop_inter1}.
\begin{proof}
The homogeneous polynomial $F$ can be written as:
\[F(X,Y) = \sum_{i +j = t\ell}{f_{ij}X^iY^j},\] with $f_{ij} \in \F_{q^m}$.
Since $F(aX,Y) = F(X,Y)$, we have: \[\sum_{i +j = t\ell}{f_{ij}X^iY^j} = \sum_{i +j = t\ell}{f_{ij}a^iX^iY^j} \cdot \]
Hence $f_{ij} = a^i f_{ij}, \forall i,j \in \mathbb{N}$ such that $ i+j = t\ell$.
As the order of $a$ is $\ell$, we have $a^i \ne 1, \forall i \in \mathbb{N}$ such that $\ell \nmid i$. Therefore $f_{ij} = 0, \forall i \in \mathbb{N}$ such that $\ell \nmid i$. So $F(X,Y) = R(X^\ell,Y^\ell)$, with $R \in \F_{q^m}[X,Y]$ an homogeneous polynomial of degree $t$.
\end{proof}

\begin{proposition}\label{prop1}
Let $\C := \C_L(\Pro^1,\PP,G)$ be an AG code as in Theorem \ref{thm}, with $\sigma$
as in (\ref{sigma1}).
Let $\tilde{P}_i = (\alpha_i^\ell:\beta_i^\ell)$ and $\tilde{G} = t \tilde{Q}
$, where either  $\tilde{Q} = ((-1)^{\ell-1} a^\frac{\ell (\ell-1)}{2} (\frac{\gamma_0}{\delta_0})^\ell : 1)$ or 
$\tilde{Q} = P_\infty$. Then $\C^\sigma = \C_{L}(\Pro^1,\tilde{\PP},\tilde{G})$, which is a GRS code.
\end{proposition}

\begin{proof}
Let $c = \big(f(P_1),f(\sigma(P_1)),\dots,f(\sigma^{\ell-1}(P_{\frac{n}{\ell}}))\big) \in \C$ such that $\sigma(c) = c$, by Lemma \ref{lem4}, $f \in L(G)$ is $\sigma$-invariant, so $f(aX,Y) = f(X,Y)$.
By Lemma \ref{lem3}, we have:
\begin{align}\label{eq1}
\frac{F(aX,Y)}{\left ( \prod\limits_{j=0}^{\ell-1}{(a\delta_jX-\gamma_jY)} \right )^{t}} & = \frac{F(X,Y)}{\left (\prod\limits_{j=0}^{\ell-1}{(\delta_jX-\gamma_jY)} \right ) ^{t}}
\end{align}
with $F \in \F_{q^m}[X,Y]$ a homogeneous polynomial of degree $t\ell$.
Moreover the support of $G$ is $\sigma$-invariant, so:
\[
\prod\limits_{j=0}^{\ell-1}{(a\delta_jX-\gamma_jY)} = \prod\limits_{j=0}^{\ell-1}{(a\delta_jX-a\gamma_jY)}= a^{\ell} \prod\limits_{j=0}^{\ell-1}{(\delta_jX-\gamma_jY)} =  \prod\limits_{j=0}^{\ell-1}{(\delta_jX-\gamma_jY)}\cdot
\]
Hence, (\ref{eq1}) becomes $F(aX,Y) = F(X,Y)$, because $\ell = \text{ord}(a)$. By Proposition \ref{prop_inter1}:
\[ F(X,Y) = R(X^\ell,Y^\ell),\] with $R \in \F_{q^m}[X,Y]$ an homogeneous polynomial of degree $t$.
  
The product $\prod\limits_{j=0}^{\ell-1}{(\delta_jX-\gamma_jY)}$ is also $\sigma$-invariant and, by Proposition \ref{prop_inter1}, we have: 
\[\prod\limits_{j=0}^{\ell-1}{(\delta_jX-\gamma_jY)} = \big(\prod\limits_{j=0}^{\ell-1}{\delta_j}\big)X^{\ell} + (-1)^{\ell} \big(\prod\limits_{j=0}^{\ell-1}{\gamma_j}\big) Y^{\ell}\cdot\]
Therefore:
\begin{equation}\label{equa_f}
f(X,Y) = \frac{R(X^\ell,Y^\ell)}{{\Big ( \big(\prod\limits_{j=0}^{\ell-1}{\delta_j}\big)X^\ell - (-1)^{\ell-1} \big(\prod\limits_{j=0}^{\ell-1}{\gamma_j}\big) Y^\ell \Big ) }^t}\cdot
\end{equation}
Forall $i \in \{1,\dots,\frac{n}{\ell}\}$, we have $f(P_i) = \tilde{f}(\tilde{P}_i)$, with:
\[\tilde{f}(X,Y) = \frac{R(X,Y)}{{\Big ( \big(\prod\limits_{j=0}^{\ell-1}{\delta_j}\big)X - (-1)^{\ell-1} \big(\prod\limits_{j=0}^{\ell-1}{\gamma_j}\big) Y \Big ) }^t}\cdot\]
and $\tilde{P}_i := (\alpha_i^\ell:\beta_i^\ell)$. We denote $\tilde{\delta} = \big(\prod\limits_{j=0}^{\ell-1}{\delta_j}\big)$ and $\tilde{\gamma} = (-1)^{\ell-1} \big(\prod\limits_{j=0}^{\ell-1}{\gamma_j}\big)$. By Lemma \ref{lem3}, we have:
\[\Bigg\{\frac{R(X,Y)}{{\Big (\tilde{\gamma} X -  \tilde{\delta} Y \Big ) }^t}~\big|~ R \in \F_{q^m}[X,Y] \text{ homogeneous polynomial of degree $t$}. \Bigg\} \cup \{0\} = L(\tilde{G}),\]
with $\tilde{G} = t (\tilde{\gamma}:\tilde{\delta}) = t \tilde{Q}$.Hence the codeword $c \in \C_{L}(\Pro^1,\tilde{\PP},\tilde{G})$.\\
If $\tilde{Q} \ne P_\infty$, then $\forall j, \delta_j \ne 0$ and  we can write:
\[\tilde{Q} = \big((-1)^{\ell-1}\prod\limits_{j=0}^{\ell-1}{\frac{\gamma_j}{\delta_j}}:1 \big).\]
Moreover we have:
\[\prod\limits_{j=0}^{\ell-1}{\frac{\gamma_j}{\delta_j}} = \prod\limits_{j=0}^{\ell-1}{a^j \frac{\gamma_0}{\delta_0}} 
= \big(\prod\limits_{j=0}^{\ell-1}{a^j}\big)(\frac{\gamma_0}{\delta_0})^\ell = a^\frac{\ell (\ell-1)}{2} (\frac{\gamma_0}{\delta_0})^\ell .
\] 

Reciprocally, for $c \in \C_{L}(\Pro^1,\tilde{\PP},\tilde{G})$ it is easy to see that we have $c = (f(P_1),\dots,f(\sigma^{\ell-1}(P_{\frac{n}{\ell}}))$, with $f$ as in (\ref{equa_f}). Then $c \in \C^\sigma$.
\end{proof}

\subsubsection{Case $\sigma$ trigonalizable over $\F_{q^m}$}

Here we consider the case where $\sigma$ is trigonalizable in $\F_{q^m}$. 
As in the previous section we only have to treat the case where $\sigma$ is upper 
triangular. So w.l.o.g one can assume that:
\begin{equation}
\begin{array}{rclc}\label{sigma2}
\sigma \colon & \Pro^1 & \to & \Pro^1\\
& (x:y) &\mapsto& (x + by:y)
\end{array}
\end{equation}
with $b \in \F^*_{q^m}$. In this case, we have $\ell = \text{ord}(\sigma) = p$.

\begin{proposition}[{\cite[Prop 4]{FOPPT16}}]\label{prop_inter2}
Let $F \in \F_q[z]$ be a polynomial of degree $deg(F) \le tp$ and $b \in \F^*_q$.
If $F(z+b) = F(z)$, then $F(z) = R(z^p-b^{p-1}z)$, with $R \in \F_q[z]$ a polynomial of degree $deg(R) \le t$.
\end{proposition}

\begin{proposition}\label{prop2}
Let $\C := \C_L(\Pro^1,\PP,G)$ be a $\sigma$-invariant AG code as in Theorem \ref{thm}, with $\sigma$
as in (\ref{sigma2}).
Let $\tilde{P}_i = (\alpha_i^p - b^{p-1}\alpha_i\beta_i^{p-1}:\beta_i^p)$ and $\tilde{G} = t(\tilde{Q})$, where either $\tilde{Q} = (( \frac{\gamma_0}{\delta_0})^p - b^{p-1}\frac{\gamma_0}{\delta_0} : 1)$ or $\tilde{Q} = P_\infty$. Then $\C^\sigma = \C_{L}(\Pro^1,\tilde{\PP},\tilde{G})$, which is a GRS code.
\end{proposition}

\begin{proof}
Let $c = \big(f(P_1),f(\sigma(P_1)),\dots,f(\sigma^{\ell-1}(P_{\frac{n}{\ell}}))\big)  \in \C$ such that $\sigma(c) = c$. By Lemma \ref{lem4}, $f$ is $\sigma$-invariant so: $f(X + bY,Y) = f(X,Y)$. 
By Lemma \ref{lem3}, we have:
\begin{align}\label{eq2}
\frac{F(X+bY,Y)}{{\Big(\prod\limits_{j=0}^{p-1}{\big(\delta_j(X + bY)-\gamma_j Y\big)}\Big)}^{t}} = \frac{F(X,Y)}{{\Big(\prod\limits_{j=0}^{p-1}{\big(\delta_jX-\gamma_jY\big)}\Big)}^{t}},
\end{align}
with $F \in \F_q[X,Y]$ an homogeneous polynomial of degree $tp$. Moreover the support of $G$ is $\sigma$-invariant, so:
\[
\prod\limits_{j=0}^{p-1}{(\delta_j(X+bY)-\gamma_jY)} = \prod\limits_{j=0}^{p-1}{(\delta_jX-(\gamma_j - b\delta_j) Y)} = \prod\limits_{j=0}^{p-1}{(\delta_jX-\gamma_jY)}.
\]
Hence, (\ref{eq2}) becomes $F(X+bY,Y)) = F(X,Y)$. If we write $z = \frac{X}{Y}$, then we have $F(z+b,1)=F(z,1)$. By Proposition \ref{prop_inter2}, we have $F(z) = R(z^p-b^{p-1}z)$, with $R \in \F_q[z]$ a polynomial of degree $deg(R) \le t$.

The product $\prod\limits_{j=0}^{p-1}{(\delta_jz-\gamma_j)}$ is also $\sigma$-invariant and, by Proposition \ref{prop_inter2}, we have:
\[\prod\limits_{j=0}^{p-1}{(\delta_jz-\gamma_j)} = \big(\prod\limits_{j=0}^{p-1}{\delta_j}\big) \big(z^p - b^{p-1}z\big) + (-1)^{p} \prod\limits_{j=0}^{p-1} \gamma_j\cdot\]
Hence:
\[f(X,Y) = \frac{R(X^p - b^{p-1}XY^{p-1},Y^p)}{{\Big(\big(\prod\limits_{j=0}^{p-1}{\delta_j}\big) \big(X^p - b^{p-1}XY^{p-1}\big) - \big((-1)^{p-1}\prod\limits_{j=0}^{p-1}{\gamma_j}\big)Y^p\Big)}^{t}}\cdot\]

The arguments to conclude this proof are the same that in Proposition \ref{prop1}.

Moreover, if $\tilde{Q} \ne P_\infty$, $\forall j, \delta_j \ne 0$ and  we can write:
\[\tilde{Q} = \big(\prod\limits_{j=0}^{p-1}{\frac{\gamma_j}{\delta_j}}:1\big).\]
We have:
\[\prod\limits_{j=0}^{p-1}{\frac{\gamma_j}{\delta_j}} = \prod\limits_{j=0}^{p-1}{(\frac{\gamma_0}{\delta_0} + jb)} = (\frac{\gamma_0}{\delta_0})^p - b^{p-1}\frac{\gamma_0}{\delta_0} \cdot 
\] 
  
\end{proof}

\subsubsection{Case $\sigma$ diagonalizable in $\F_{q^{2m}} \backslash \F_{q^m} $}\label{section3.2.3} We suppose that $\sigma = \rho \circ \sigma_d 
\circ \rho^{-1}$ with $\sigma_d$ diagonal in GL$_2(\F_q^{2m})$ and $\rho \in \PGL_2(F_{q^{2m}})$. We want to extend the code $\C$ defined on $\F_{q^m}$ to the field $\F_{q^{2m}}$. So we consider the set $\text{Span}_{\F_{q^{2m}}}<\C>$, i.e. $\C \otimes \F_{q^{2m}}$.

The order $\ell$ of $\sigma_d := \begin{pmatrix}
\alpha & 0\\
0 & \alpha^q
\end{pmatrix}$ is the order of $\alpha \in \F_{q^{2m}}$, so $\ell \mid (q^{2m} -1)$. 
Since $q := p^s$, where $s \in \mathbb{N}^*$, we have $\ell \mid (p^{s2m} -1)$ and 
so $p \nmid \ell$. By Lemma \ref{lem1}, $\fold_{\sigma} = \inv_{\sigma}$, and so we have the following diagram:

\[
\xymatrix{
\relax \C \otimes \F_{q^{2m}} = \{\Ev_{\PP}(f) | f \in L_{\F_{q^{2m}}}(G)\} \ar[r]^-{\fold_{\sigma_d}} & \inv_{\sigma_d}(\C \otimes \F_{q^{2m}}) \\
 \C = \{\Ev_{\PP}(f) | f \in L_{\F_{q^m}}(G)\} \ar[r]^-{\fold_{\sigma}} \ar@{^{(}->}[u]^{\text{Subfield Subcode}} & \inv_{\sigma}(\C) \ar@{^{(}->}[u]_{\text{Subfield Subcode}}
} \]

By Section \ref{diag}, the code $\inv_{\sigma}(\C \otimes {\F_{q^{2m}}})$ is 
a GRS code.
Since the application $\fold_{\sigma}$ is $\F_q$-linear and $\C \otimes \F_
{q^{2m}}$ has, by definition, a basis in $\F_{q^m}^n$, the code $\inv_{\sigma}
(\C \otimes \F_{q^{2m}})$ also has a basis in $\F_{q^m}^n$. Therefore, the subfield subcode on $\F_{q^m}$ of the GRS code $\inv_{\sigma}(\C \otimes \F_{q^{2m}})$ is a GRS 
code.

\section{Security of $\sigma$-invariant Alternant Codes}\label{sec4}

In this section, we study the security of keys of the McEliece scheme based alternant codes, with $\sigma$ an automorphism of the projective line acting on it.
We consider an automorphism $\sigma \in \PGL_2(\F_{q^m})$, and a $\sigma$-invariant alternant code $\Alt(\Pro^1,\PP,G) = \C_L(\Pro^1,\PP,G)^\perp \cap \F_q^n$, with $\PP$ a $\sigma$-invariant set of distinct points of $\Pro^1$ defined as (\ref{support}) and $G$ a $\sigma$-invariant divisor of $\Pro^1$ defined as (\ref{div}). 
By Section \ref{sec_inv} and Corollary \ref{coro}, the invariant code $\Alt(\Pro^1,\PP,G)^\sigma$ is an alternant code $\Alt(\Pro^1,\tilde{\PP},\tilde{G})$, with smaller parameters. 

Here we assume that it is possible to recover $\tilde{\PP}$ and $\tilde{G}$ from a generator matrix of the code $\Alt(\Pro^1,\tilde{\PP},\tilde{G})$. This can be done by a brute force attack, if the parameters of $\Alt(\Pro^1,\tilde{\PP},\tilde{G})$ are smaller enough. Otherwise, we assume that an algebraic attack, proposed in \cite{FOPT10}, can recover $\tilde{\PP}$ and $\tilde{G}$. 

We will show that thanks to the knowledge of $\tilde{\PP}$ and $\tilde{G}$ we are able to recover $\PP$ and $G$. We already know that there is a link between the form of $\tilde{\PP}$ and $\tilde{G}$ of the invariant code and the form of $\PP$ and $G$ of the original alternant code. The link is described by Propositions \ref{prop1} and \ref{prop2} of the previous section.

Later on, we denote $\tilde{\PP} := \big \{(\tilde\alpha_i : 1) ~|~ i \in \{1,\dots,\frac{n}{\ell} \} \big \}$ the support of the invariant code. All the $\tilde\alpha_i$ are known. 
We denote by $(\alpha_{i,j}:1)$ for $i \in \{1,\dots,\frac{n}{\ell}\}$ and $j \in \{0,\dots,\ell-1\}$ the elements of $\PP$. We assume that $G$ is 
constructed from one rational point $Q$, ie: $G = t \sum_{j = 0}^{\ell-1}{\sigma^j(Q)}$, with $ \sigma^j(Q) := (\gamma_j:\delta_j)$, for all $j \in \{0,\dots,\ell-1\}$. The result remains true for the general case.

\subsection{Recover the divisor and guess the support}\label{sec_div}

As previously, we know three cases are possibles: $\sigma$ is diagonalizable over $\F_{q^m}$, $\sigma$ is trigonalizable, or $\sigma$ is diagonalizable over $\F_{q^{2m}}$. In this section we treat the two first cases, the third case will be treated at the end of Section \ref{sec4}. The order $\ell$ of $\sigma$ is known, hence we know the form of $\sigma$.

\subsubsection{Case $\sigma$ diagonalizable over $\F_{q^m}$}
In this case, we recall that the form is:
\begin{equation}
\begin{array}{rclc}\label{sigma_bis_1}
\sigma \colon & \Pro^1 &\to &\Pro^1 \\ 
&(x:y)& \mapsto &(ax:y),
\end{array}
\end{equation}
with $a \in \F_{q^m}^\times$, an $\ell$-th root of unity. There exist only $\varphi(\ell) < n$ possibilities for $a$ , where $\varphi$ the Euler's phi function, hence we are able to test all the possibilities. W.l.o.g we assume for now that we know the element $a$. The first step is to recover $G$ from $\tilde{G}$.
By Proposition \ref{prop1}, we know that $\tilde{G} = t\tilde{Q}$, where either $\tilde{Q} = ((-1)^{\ell-1} a^\frac{\ell (\ell-1)}{2} (\frac{\gamma_0}{\delta_0})^\ell : 1)$ or 
$\tilde{Q} = P_\infty$.
Since we know $a$, we can recover the support of $G$ thanks to the support $\tilde{Q}$ of $\tilde{G}$.
\begin{remark}
For all $i \in \{0,\dots,\ell-1\}$, we have $\tilde{Q} = ((-1)^{\ell-1} \sigma^i(a)^\frac{\ell (\ell-1)}{2} (\frac{\gamma_i}{\delta_i})^\ell : 1)$ . We denote $A := \{\sigma^i(a)~|~i \in \{0,\dots,\ell-1\}\}$ and then from every $a \in A$ we are able to recover the support of $G$. The set $A$ is exactly the set of roots of the unity.
\end{remark}

\begin{algorithm}\label{algo1}
\SetAlgoVlined
\caption{Recover the divisor in the case $\sigma$ diagonalizable in $\F_{q^m}$}
\SetKwInOut{Input}{Input} \SetKwInOut{Output}{Output} 
\Input{The divisor $\tilde{G}$ of the invariant code $\Alt(\PP,G)^\sigma$.
}
\Output{Return the support $G$}
{$ a \leftarrow$  a primitive $\ell$-th root of $\F_{q^m}$\\
\uIf(\eolcomment{$\tilde{G} = t * \tilde{Q}$, with $\tilde{Q} = (\tilde\gamma : 1)$}){$\tilde{Q} \ne P_\infty$}{
$\Gamma \leftarrow \roots(a^{\frac{\ell(\ell-1)}{2}}X^\ell - \tilde{\gamma} )$ \eolcomment{$a \in A $}\\
$G' \leftarrow t\sum\limits_{\gamma \in \Gamma}{(\gamma : 1)}$\\
}
\Else{
$G'= t P_\infty$\\
}
{\bf{return}} $G$.
}
\end{algorithm}

The second step is to recover a support $\PP'$ such as $\Alt(\PP',G) = \Alt(\PP,G)$ . By Proposition \ref{prop1}, we know that a point $P = (x : y)$ in $\PP$ satisfies:
\begin{equation}\label{equa_P1}
\begin{cases}
x^\ell - \tilde\alpha_i = 0\\
y^\ell - \tilde\beta_i = 0, 
\end{cases}
\end{equation}
for some $i \in \{1,\dots,\frac{n}{\ell}\}$ such that $(\tilde\alpha_i:\tilde\beta_i) = \tilde{P}_i$. Since we know $\tilde\PP$, we are able to recover all elements of $\PP$ but as an unordered set. 
We choose one solution $(\alpha_i, \beta_i)$ of (\ref{equa_P1}) for each $i \in \{1,\dots,\frac{n}{\ell}\}$ and we choose $a \in A$, then the set:
\begin{equation}\label{equa_P_prim}
\PP' := \Big\{ \big(a^j\frac{\alpha_i}{\beta_i} :1 \big)~\big|~j \in \{0,\dots,\ell-1\}, i \in \{1,\dots,\frac{n}{\ell}\}\Big\}
\end{equation}
is a support such as $\Alt(\PP',G)$ is a permutation of the code $\Alt(\PP,G)$. For each choice of set of solutions $S := \{ (\alpha_i,\beta_i) ~|~ i \in \{1,\dots,\frac{n}{\ell} \} \}$ and each choice of $a \in A$, we have a different support $\PP'$. In Section \ref{sec_perm} we give an algorithm to find a good choice for $S$ and $a$ and hence the permutation between $\Alt(\PP',G)$ and $\Alt(\PP,G)$.

\subsubsection{Case $\sigma$ trigonalisable over $\F_{q^m}$}
In the case where $\sigma$ is trigonalisable, it is more complicated to know exactly $\sigma$. In this case we know that $\sigma$ has the following form:
\begin{equation}
\begin{array}{rclc}\label{sigma_bis_2}
\sigma \colon & \Pro^1 &\to &\Pro^1 \\ 
&(x:y)& \mapsto &(x + by:y),
\end{array}
\end{equation}
with $b \in \F_{q^m}^\times$. Here the order of  $\sigma$ is $\ell = p :=$ Char($\F_{q^m}$) and the first step is to recover $b$.

\begin{lemma}\label{lem2}
$b$ is a root of the polynomial:
\[P_b := \gcd \Big (\big\{\Res_X( X^p - Y^{p-1}X - \tilde\alpha_i,~ X^{q^m}-X) ~|~ i \in \{1,\dots,\frac{n}{\ell}\} \big\},~ Y^{q^m}-Y \Big),\] where $\Res_X(P,Q)$ denotes the 
resultant of the two polynomials $P$ and $Q$ with respect to $X$.
\end{lemma}

\begin{proof}
By Proposition \ref{prop2}, $b$ is a root of the polynomial 
$\alpha_{i,j}^p - Y^{p-1}\alpha_{i,j} - \tilde\alpha_i \in \F_{q^m}[Y]$ for all 
$i \in \{1,\dots,\frac{n}{p}\}$ and $j \in \{0,\dots,p-1\}$.
As $\alpha_{i,j} \in \F_{q^m}^\times$ for all $i,j$, we can also write that $b$ is a 
root of the polynomial $\Res_X( X^p - Y^{p-1}X - \tilde\alpha_i,~ X^{q^m}-X) \in \F_{q^m}[Y]$ for all $i \in \{1,\dots,\frac{n}{p}\}$. 
\end{proof}

All the elements of the orbit of $b$ under the action of $\sigma$ are roots of the polynomial $P_b$ defined previously, i.e: the elements of the set $B := \{b, 2b,\dots, (\ell-1)b \}$. In practice, and according to computer aided experiment, the degree of the polynomial $P_b$ is $\ell$ and the set $B$ is exactly the set of its roots. Then there exist only $\ell < n$ possibilities for $b$, so we assume for now that we know the element $b$.

The second step is to recover the divisor $G$ from $\tilde{G}$. By Proposition \ref{prop2}, we know that $\tilde{G} = t\tilde{Q}$, where either $\tilde{Q} = (( \frac{\gamma_0}{\delta_0})^p - b^{p-1}\frac{\gamma_0}{\delta_0} : 1)$ or $\tilde{Q} = P_\infty$.. Since we know $a$, we can recover the support of $G$ thanks to the support $\tilde{Q}$ of $\tilde{G}$.

Since we know $b$, we can recover the support of $G$ thanks to the support $\tilde{Q}$ of $\tilde{G}$.

\begin{algorithm}\label{algo2}
\SetAlgoVlined
\caption{Recover the divisor in the case $\sigma$ trigonalizable over $\F_{q^m}$}
\SetKwInOut{Input}{Input} \SetKwInOut{Output}{Output} 
\Input{The divisor $\tilde{G}$ of the invariant code $\Alt(\PP,G)^\sigma$.
}
\Output{Return $G$ and $B$ the set of possible values of $b$ in (\ref{sigma_bis_2})}
\inlcomment{Recover $B := \{b, 2b, \dots, (p-1)b\}$}
\eolcomment{$\tilde{\PP} = \{ \tilde{\alpha_i} ~|~ i \in \{1,\dots,\frac{n}{\ell}\}\}$}\\
{$R \leftarrow \gcd \big (\{\Res_X( X^p - Y^{p-1}X - \tilde\alpha_i,~ X^{q^m}-X) ~|~ i \in \{1,\dots,\frac{n}{p}\} \},~ Y^{q^m}-Y \big)$}\\
$B \leftarrow \roots(R)$\\
\inlcomment{Recover $G'$ from $\tilde{G}$}\\
\eIf(\eolcomment{$\tilde{G} = t * \tilde{Q}$, with $\tilde{Q} = (\tilde\gamma : 1)$}){$\tilde{Q} \ne P_\infty$}{ 
$\Gamma \leftarrow \roots(X^p - b^{p-1}X - \tilde{\gamma})$ \eolcomment{$b \in B $}\\
$G' \leftarrow t\sum\limits_{\gamma \in \Gamma}{(\gamma : 1)}$\\
}
{$G'= t P_\infty$}
{\bf{return}} $G', B$
\end{algorithm}

\begin{proposition}
Algorithm \ref{algo2} finds the set $B$ and the divisor $G$ in $\OO(n(q^m + p)^{\omega + 1})$ operations in $\F_{q^m}$, where $\omega$ is the exponent of the linear algebra.  
\end{proposition}

\begin{proof}
We only prove the cost of the algorithm, its correctness is a consequence of Lemma \ref{lem2} and Lemma \ref{lem3}.

The resultant of two polynomials can be computed with an 
"evaluation-interpolation" me\-thod, which reduces to compute determinants of scalar matrices and one interpolation. Here the degree of the resultant that we must compute is at most $(q^m +p)(p-1)$, so we must compute $(q^m +p)(p-1) + 1$ determinants with scalar coefficients. The only polynomial to evaluate here is $Y^{(p-1)}$, so the cost of the evaluation is $\OO((q^m+p)(p-1)\log_2(p-1))$. Computing determinants costs $ \OO((q^m +p)^\omega ((q^m +p)(p-1)))$ operations, where $\omega$ is the exponent of the cost of linear algebra. 
Then the interpolation cost $(q^m +p)^2(p-1)^2$ operations, using Lagrange interpolation, but this is negligible behind the cost of the previous determinants.\\
With Euclid Algorithm we can compute the gcd of two polynomials in $\F_{q^m}[Y]$ of degree at most $(q^m +p)(p-1)$ in $\OO\big((q^m +p)p)^2\big)$ operations in $\F_{q^m}$. We compute at most $\frac{n}{p}$ resultants and gcd, so the cost of the first step is $\OO(n(q^m +p)^{\omega + 1})$ operations in $\F_{q^m}$.

The second step is negligible behind the first step event if we use Berlekamp algorithm.
\end{proof}

The third step is to recover a support $\PP'$ such as $\Alt(\PP',G) = \Alt(\PP,G)$ . By Proposition \ref{prop2}, we know that a point $P = (x : y)$ in $\PP$ satisfies:
\[\begin{cases*}
x^p - b^{p-1}x - \tilde\alpha_i = 0\\
y^p - \tilde\beta_i = 0 
\end{cases*}
\]
for $ i \in \{1,\dots,\frac{n}{\ell}\}$, such that $(\tilde\alpha_i:\tilde\beta_i) = \tilde{P}_i$. Since we know $\tilde\PP$, we are able to recover all elements of $\PP$ but as an unordered set. Hence we know a support $\PP'$ such as $\Alt(\PP',G)$ is a permutation of the code $\Alt(\PP,G)$.

\subsection{Recover the permutation}\label{sec_perm}
At this point the problem is to recover the permutation between $\Alt(\PP',G)$ and $\Alt(\PP,G)$. Let $Gen_{\Alt(\PP,G)}$ be a generator matrix of the code $\Alt(\PP,G)$, and $H_{\Alt(\PP',G)}$ be a parity check matrix of the code $\Alt(\PP',G)$, the permutation between $\Alt(\PP,G)$ and $\Alt(\PP',G)$ is represented by matrix $\Pi$ such that \begin{equation}\label{equ1}
Gen_{\Alt(\PP,G)} \Pi H_{\Alt(\PP',G)}^T = 0.
\end{equation}

If we have no assumption on the permutation between $\PP'$ and $\PP$, to find the permutation $\Pi$ we must resolve a linear system with $n^2$ unknowns, while \eqref{equ1} is a system of  $k(n-k)$ linear equations which is not enough to find an unique solution. 

Now we assume that we made the good choice for $a \in A$ (or $b \in B$). Then the permutation matrix $\Pi$ has the following form:

\begin{equation}\label{equ_P}
\Pi = 
\left(
\raisebox{0.5\depth}{%
\xymatrixcolsep{1ex}
\xymatrixrowsep{1ex}
\xymatrix{
\sum\limits_{i = 1}^{\ell}{x_{1,i}J^i} \ar@{.}[dr]& \text{\huge (0)}\\
\text{\huge (0)}&\sum\limits_{i = 1}^{\ell}{x_{\frac{n}{\ell},i}J^i} 
}%
}
\right)
\text{ where } 
J := \left(
\raisebox{0.5\depth}{%
\xymatrixcolsep{1ex}
\xymatrixrowsep{1ex}
\xymatrix{
0 \ar @{.}[rr] \ar@{.}[dddrrr]& & 0 \ar@{.}[dr] & 1\\
1 \ar @{-}[ddrr] &&& 0 \ar@{.}[dd]\\
0 \ar@{.}[d] \ar@{.}[dr] &&&\\
0 \ar@{.}[r]&0 &1 & 0 \\
}%
}
\right)\cdot
\end{equation}

$J$ is an $\ell \times \ell$ matrix, and $x_{j,i} \in \{0,1\}$ are unknowns, for $j \in \{1,\dots,\frac{n}{\ell}\}$ and $i \in \{1,\dots,\ell\}$. With this form we have $n$ unknowns. Assume that $n - k \le \frac{n}{2}$, then $n \le (n-k)k$. In this case, we can hope to find a unique solution for $\Pi$. In all our computer aided experiments we got a unique solution for $\Pi$. If the choice for $a \in A$ (or $b \in B$) is wrong, then there was no solution for the system in all our experiments.

We present an algorithm to recover the permutation matrix $\Pi$ and so the good choice for $a \in A$ (or $b \in B$). The algorithm is only written for the first case, where $\sigma$ is diagonalizable, but the other case is similar. 

\begin{algorithm}\label{Algo}
\SetAlgoVlined
\caption{Recover the support}
\SetKwInOut{Input}{Input} \SetKwInOut{Output}{Output} 
\Input{A generator matrix of a quasi-cyclic alternant code: $Gen_{\Alt(\PP,G)}$, 
the divisor $G$, and the support $\tilde{\PP}$ of the invariant code.
}
\Output{Returns FALSE if no solution is found. Else, returns TRUE and $\PP'$ such that $\Alt(\PP',G) = \Alt(\PP,G)$}
\For{$i \in \{1,\dots,\frac{n}{\ell}\}$}{
$\alpha_i \leftarrow \roots(x^\ell-\tilde{\alpha_i})[1]$ \eolcomment{cf \eqref{equa_P_prim}} \\
$\beta_i \leftarrow \roots(y^\ell-\tilde{\beta_i})[1]$ \\
}
\ForAll{$a \in A$}{ 
\inlcomment{Guess $\PP'$}\\
$\PP' \leftarrow \{ (a^j \frac{\alpha_i}{\beta_i} : 1) ~|~ j \in \{0 \dots \ell-1\}, i \in \{1\dots \frac{n}{\ell}\}  \}$\\
$\C \leftarrow \Alt(P',G) $\\
\uIf{$\C = \Alt(\PP,G)$ }{
{\bf{return}} TRUE, $\PP'$
}
\Else{
{$H \leftarrow $ ParityCheckMatrix($\C$)}\\
{$S \leftarrow$ \bf{solve}}($Gen_{\Alt(\PP,G)}\Pi H^T = 0$, with $\Pi$ a permutation matrix of the form \eqref{equ_P})\\
\If{$|S| = 1$}{
{\bf{return}} TRUE, $(\PP' * \Pi)$
}
}
}
{\bf{return}} FALSE
\end{algorithm}

\begin{proposition}
Algorithm \ref{Algo} finds a support $\PP'$ such that $\Alt(\PP',G) = \Alt(\PP,G)$
in $\OO(\ell n^{2}(n-k)k )$ operations in $\F_{q^m}$,
where $n$ is the length of $\Alt(\PP,G)$, $k$ is the dimension of $\Alt(\PP,G)$, $\ell$ is the order of $\sigma$. 
\end{proposition}

\begin{proof}
We only prove the cost of the algorithm.
We must to resolve a linear system of $(n-k)k$ equations with $n$ unknowns, this is possible in $\OO(n^{2}(n-k)k)$ operations in $\F_{q^m}$.
This step is repeated at most $\ell$ times so the cost is in $\OO(\ell n^{2}(n-k)k )$ operations in $\F_{q^m}$.
\end{proof}

In order to give practical running times for this part of the attack, we implemented Algorithm \ref{Algo} in \textsc{Magma} \cite{Magma}. The platform used in the experiments is a 2.27GHz Intel\textsuperscript{\textregistered} Xeon\textsuperscript{\textregistered} Processor E5520. For each set of parameters, we give the average time obtained after 10 tests. In the following table we use notation: 

\begin{itemize}
  \item $m$ : extension degree of the field of definition of
    the support and divisor over $\F_q$
  \item $n$ length of the quasi--cyclic code
  \item $k$ dimension of the quasi--cyclic code
  \item $\ell$ denotes the order of quasi--cyclicity of the code
  \item $w_{ISD}$ denotes the logarithm of the
    work factor for message recovery
    attacks. It is computed using {\bf CaWoF} library \cite{CT16}.
\end{itemize}

\begin{center}
  \begin{tabular}[!h]{|c|c|c|c|c|c|c|c|c|c|}
    \hline
    $q$ & $ m$ & $n$ & $ k$ & $ \ell$  & $w_{ISD}$ & Algorithm \ref{Algo}\\ 
    \hline
2 & 12 & 3600 & 2825 & 3 & 129 &  1659 s ($\approx$ 27 min) \\ 
 \hline 
2 & 12 & 3500 & 2665 & 5 & 130 & 2572 s ($\approx$ 42 min) \\ 
 \hline 
2 & 12 & 3510 & 2579 & 13 & 132 & 8848 s ($\approx$ 2h27) \\ 
 \hline 
  \end{tabular}
\end{center}

\subsection{Case $\sigma$ diagonalizable in $\F_{q^{2m}} \backslash \F_{q^m} $}
In this case, we recall that $\sigma = \rho \circ \sigma_d \circ \rho^{-1}$ with $\rho \in \GL_2(\F_{q^{2m}})$ and:
\[
\begin{array}{rclc}
\sigma_d \colon & \Pro^1 &\to &\Pro^1 \\ 
&(x:y)& \mapsto &(\alpha x:\alpha^q y),
\end{array}
\]
where $\alpha \in \F_{q^{2m}}$ is an $\ell$-th root of unity. As $\sigma_d$ is 
diagonal in $\F_{q^{2m}}$, we can recover a support $\PP'$ and a divisor $G'$ in $\F_{q^{2m}}
$, using the same method as in Sections \ref{sec_div} and \ref{sec_perm}. 

Now we want to recover a support $\PP$ and a divisor $G$ in $\F_{q^m}$. We consider $\pi_\alpha := X + a X + b$ the minimal polynomial of $\alpha$, with $a,b \in \F_{q^m}$. 
Then:
\[
M_{\sigma_d} = \begin{pmatrix}
\alpha & 0 \\
0 & \alpha^q
\end{pmatrix}
\sim
\begin{pmatrix}
0 & -b \\
1 & -a
\end{pmatrix} = M_{\sigma'}
\]
and there exist $\rho' \in \GL_2(\F_{q^{2m}})$ such that $\sigma_d = \rho' \circ \sigma' \circ \rho'^{-1}$, where $\sigma'$ is the element of $\PGL_2(\F_{q^m})$ 
associated to $M_{\sigma'}$. We can assume that $\sigma = \sigma'$, then we want to 
recover $\rho'$. Thanks to Section \ref{sec_perm}, we know $\alpha$ and it is easy 
to compute $a$ and $b$. To recover $\rho'$ it suffices to diagonalize the matrix $M_
{\sigma'}$. From $\rho'$ and a support $\PP'$ and a divisor $G'$ in $\F_{q^{2m}}$, we can recover a support $\PP = \rho'^{-1}(\PP')$ and a divisor $G = \rho'^{-1}(G')$
in $\F_{q^m}$.

\section{Conclusion}

To summarise, we showed that the key security of compact McEliece schemes based on 
alternant codes with some induced permutation is not better than the key security of 
the short code obtained from the invariant operation. A similar result was showed for permutations induced by the affine group with the folded code. Our new approach simplifies the reduction and extend the result to the projective linear
group. Moreover, we present a simpler lifting from the invariant code to recover the original code, with method from linear algebra. 

Another kind of quasi-cyclic alternant codes could be obtained from the 
action of the \textit{semilinear projective group} on the support. By semilinear projective 
group, we mean transformation of the form: 
$x \mapsto \frac{ax^{q^i} + b}{cx^{q^i} + d}$, with $a,b,c,d \in \mathbb{F}_{q^m}^n$. These transformations induce a permutation on the alternant code $\C \cap \F_q^n$ 
but not on the GRS code $\C$. So we cannot use the same property of the invariant of 
a GRS code to study this kind of quasi-cyclic alternant code.

We can notice that key-recovery is generally more expensive than message recovery. With a good 
choice of parameters it is still possible to construct quasi-cyclic codes with high complexity of key recovery attack on the invariant code. 
 
\subsubsection*{Acknowledgements}
This work is partially supported by a DGA-MRIS scholarship, by French 
ANR-15-CE39-0013-01 "Manta", and by European grant CORDIS ICT-645622 "PQCrypto".
We would like to thank J.P. Tillich and J. Lavauzelle for helpful discussions, and A. Couvreur for many valuable comments on the preliminary versions of this paper.

\bibliography{biblio}
\bibliographystyle{amsplain}
\end{document}